\pdfoutput=1
\documentclass[runningheads]{llncs}
\usepackage{amssymb}
\usepackage{amsmath}
\usepackage{graphicx}
\usepackage{verbatim}
\usepackage{color}
\usepackage{ae}
\usepackage[ruled]{algorithm}
\usepackage{algorithmicx}
\usepackage{algpseudocode}
\usepackage[breaklinks,bookmarks=false]{hyperref}

\usepackage{amssymb}
\usepackage{amsmath}
\usepackage{graphicx}
\usepackage{verbatim}
\usepackage{color}
\usepackage{ae}
\usepackage[ruled]{algorithm}
\usepackage{algorithmicx}
\usepackage{algpseudocode}
\usepackage{hyperref}

\newcommand{\twodots}{\mathinner{\ldotp\ldotp}}

\newcommand{\proc}[1]{\textnormal{\scshape#1}}

\DeclareMathOperator{\PREF}{PREF}

\begin{document}

\title{Tying up the loose ends in fully LZW-compressed pattern matching\thanks{Supported by MNiSW grant number N~N206 492638, 2010--2012 and START scholarship from FNP.}}
\titlerunning{Fully LZW-compressed pattern matching} 

\author{Pawe\l{} Gawrychowski}
\institute{
	Institute of Computer Science,\\
	University of Wroc{\l}aw,\\
	ul. Joliot-Curie 15, 50--383 Wroclaw,
	Poland \\
	\email{gawry@cs.uni.wroc.pl}
} 

\maketitle

\begin{abstract}
We consider a natural generalization of the classical pattern matching problem: given compressed representations of a pattern $p[1\twodots M]$ and a 
text $t[1\twodots N]$ of sizes $m$ and $n$, respectively, does $p$ occur in $t$? We develop an optimal linear time solution for the case when $p$ and $t$ are compressed using the LZW method. This improves the previously known $\mathcal{O}((n+m)\log(n+m))$ time solution of G\k{a}sieniec and Rytter~\cite{RytterGasieniec}, and essentially closes the line of research devoted to studying LZW-compressed exact pattern matching.

\end{abstract}

\section{Introduction}

One of the most natural problems concerning processing information is {\it pattern matching}, in which we are given a pattern $p[1\twodots M]$ and a text $t[1\twodots N]$, and have to check if there is an occurrence of $p$ in $t$. Although many very efficient (both from a purely theoretical and more practically oriented) solution to this problem are known~\cite{MP,RealGalil,KMP,ConstGalil,Breslauer,BM}, most data is archived and stored in a compressed form. This suggest an intriguing research direction: if the text, or both the pattern and the text, are given in their compressed representations, do we really need to decompress them in order to detect an occurrence? If just the text is compressed, this is the {\it compressed pattern matching} problem. For Lempel-Ziv-Welch compression, used for example in Unix \texttt{compress} utility and GIF images, Amir {\it et al.} introduced two algorithm with respective time complexities $\mathcal{O}(n+M^2)$ and $\mathcal{O}(n\log M+M)$, where $n$ is the compressed size of the text. The pattern preprocessing time was then improved~\cite{Kosaraju} to get $\mathcal{O}(n+M^{1+\epsilon})$ time complexity. In a recent paper~\cite{GawrychowskiLZW} we proved that in fact a $\mathcal{O}(n+M)$ solution is possible, as long as the alphabet consists of integers which can be sorted in linear time. A more general problem is the {\it fully compressed pattern matching}, where both the text and the pattern are compressed.
This problem seems to be substantially more involved than compressed pattern matching, as we cannot afford to perform any preprocessing for every possible prefix/suffix of the pattern, and such preprocessing is a vital ingredient of any efficient pattern matching algorithm known to the author.
Nevertheless, G\k{a}sieniec and Rytter~\cite{RytterGasieniec} developed a $\mathcal{O}((n+m)\log(n+m))$ time algorithm for this problem, where $n$ and $m$ are the compressed sizes of the text and the pattern, respectively. 

In this paper we show that in fact an optimal linear time solution is possible for fully LZW-compressed pattern matching. The starting point of our algorithm is the $\mathcal{O}(n+M)$ time algorithm~\cite{GawrychowskiLZW}. Of course we cannot afford to use it directly as $M$ might be of order $m^2$. Nevertheless, we can apply the method to a few carefully chosen fragments of the pattern. Then, using those fragments we try to prune the set of possible occurrences, and verify them all at once by a combined usage of the so-called $\PREF$ function and iterative compression of the (compressed) pattern similar to the method of~\cite{RytterGasieniec}. The chosen fragments correspond to the most complicated part of the pattern, in a certain sense. If there is no such part, we observe that the pattern is periodic, and a modification of the algorithm from~\cite{GawrychowskiLZW} can be applied. To state this modification, and prove its properties, we briefly review the algorithm from~\cite{GawrychowskiLZW} in the next section. While the modification itself might seem simple, we would like to point out that it is nontrivial, and we need quite a few additional ideas in order to get the result.

\section{Preliminaries}

We consider strings over finite alphabet $\Sigma$ (which consists of integers which can be sorted in linear time, namely $\Sigma=\{1,2,\ldots,(n+m)^c \}$) given in a Lempel-Ziv-Welch compressed form, where a string is represented as a sequence of {\it codewords}. Each codeword is either a single letter or a previously occurring codeword concatenated with a single character. This additional character is not given explicitly: we define it as the first character of the next codeword, and initialize the set of codewords to contain all single characters in the very beginning. The resulting compression method enjoys a particularly simple encoding/decoding process, but unfortunately requires outputting at least $\Omega(\sqrt{N})$ codewords (so, for example, we are not able to achieve an exponential compression possible in the general Lempel-Ziv method). Still, its simplicity and good compression ratio achieved on real life instances make it an interesting model to work with. For the rest of the paper we will use LZW when referring to Lempel-Ziv-Welch compression.

We are interested in a variation of the classical pattern matching problem: given a pattern $p[1\twodots M]$ and a text $t[1\twodots N]$, does $p$ occur in $t$? We assume that both $p$ and $t$ are given in LZW compressed forms of size $n$ and $m$, respectively, and wish to achieve a running time depending on the size of the compressed representation $n+m$, not of the original lengths $N$ and $M$. If the pattern does occur in the text, we would like to get the position of its first occurrence. We call such problem the {\it fully compressed pattern matching}.

A closely related problem is the {\it compressed pattern matching}, where we aim to detect the first occurrence of an uncompressed pattern in a compressed text. In a previous paper~\cite{GawrychowskiLZW}, we proved that this problem can be solved in deterministic linear time. This $\mathcal{O}(n+M)$ time algorithm will be our starting point. Of course we cannot directly apply it as $M$ might be of order $m^2$. Nevertheless, a modified version of this solution will be one of our basic building bricks. In order to state the modification and prove its correctness we briefly review the idea behind the original algorithm in the remaining part of this section. First we need a few auxiliary lemmas. A {\it period} of a word $w$ is an integer $0<d<|w|$ such that $w[i]=w[i+d]$ whenever both letters are defined.

\begin{lemma}[Periodicity lemma]\label{lemma:periodicity}
If both $d$ and $d'$ are periods of $w$, and $d+d'\leq |w|+\gcd(d,d')$, then $\gcd(d,d')$ is a period as well.
\end{lemma}

Using any linear time suffix array construction algorithm and fast LCA queries~\cite{BenderLCA} we get the following.

\begin{lemma}\label{lemma:equality}
Pattern $p$ can be preprocessed in linear time so that given any two fragments $p[i\twodots i+k]$ and $p[j\twodots j+k]$ we can find their longest common prefix (suffix) in constant time.
\end{lemma}

\begin{lemma}\label{lemma:longest suffix}
Pattern $p$ can be preprocessed in linear time so that given any fragment $p[i\twodots j]$ we can find its longest prefix which is a suffix of the whole pattern in constant time, assuming we know the (explicit or implicit) vertex corresponding to $p[i\twodots j]$ in the suffix tree.
\end{lemma}

A {\it border} of a word $w$ is an integer $0<b<|w|$ such that $w[1\twodots b]=w[|w|-b\twodots |w|]$. By applying the preprocessing from the Knuth-Morris-Pratt algorithm for both the pattern and the reversed pattern we get the following.

\begin{lemma}\label{lemma:borders preprocessing}
Pattern $p$ can be preprocessed in linear time so that we can find the border of each its prefix (suffix) in constant time.
\end{lemma}

A {\it snippet} is any substring $p[i\twodots j]$ of the pattern, represented as a pair of integers $(i,j)$. If $i=1$ we call it a {\it prefix snippet}, and if $j=|p|$ a {\it suffix snippet}. An {\it extended snippet} is a snippet for which we also store the corresponding vertex in the suffix tree 
(built for the pattern) and the longest suffix which is a prefix of the pattern. A sequence of snippets is a concatenation of a number of substrings of the pattern.

A high-level idea behind the linear time LZW-compressed pattern matching is to first reduce the problem to pattern matching in a sequence of extended snippets. It turns out that if the alphabet is of constant size, the reduction can be almost trivially performed in linear time, and for polynomial size integer alphabets we can apply more sophisticated tools to get the same complexity. Then we focus on pattern matching in a sequence of snippets. The idea is to simulate the well-known Knuth-Morris-Pratt algorithm while operating on whole snippets instead of single characters using Lemma~\ref{lemma:concatenation occurrence} and Lemma~\ref{lemma:extending long border}.

\begin{lemma}\label{lemma:concatenation occurrence}
Given a prefix snippet and a suffix snippet we can detect an occurrence of the pattern in their concatenation in constant time.
\end{lemma}

\begin{lemma}\label{lemma:extending long border}
Given a prefix snippet $p[1\twodots i]$ and a snippet $p[j\twodots k]$ we can find the longest long border $b$ of $p[1\twodots i]$ such that $p[1\twodots b]p[j\twodots k]$ is a prefix of the whole $p$ in constant time, where a long border is $b\geq\frac{i}{2}$ such that $p[1\twodots b]=p[i-b+1\twodots i]$.
\end{lemma}

During the simulation we might create some new snippets, but they will be always either prefix snippets or {\it half snippets} of the form $p[\frac{i}{2}\twodots i]$. All information required to make those snippets extended can be precomputed in a relatively straightforward way using $\mathcal{O}(m)$ time.

The running time of the resulting procedure is as much as $\Theta(n\log m)$, though. To accelerate it we try to detect situations when there is a long snippet near the beginning of the sequence, and apply Lemma~\ref{lemma:lever occurrence} and Lemma~\ref{lemma:lever prefix} to quickly process all snippets on its left.

\begin{lemma}\label{lemma:lever occurrence}
Given a sequence of extended snippets $s_1 s_2 \ldots s_i$ such that $|s_i|\geq 2\sum_{j<i}|s_j|$, we can detect an occurrence of $p$ in $s_1 s_2 \ldots s_i$ in time $\mathcal{O}(i)$.
\end{lemma}

\begin{lemma}\label{lemma:lever prefix}
Given a sequence of extended snippets $s_1 s_2 \ldots s_i$ such that $|s_i|\geq 2\sum_{j<i}|s_j|$, we can compute the longest prefix of $p$ which is a suffix of $s_1 s_2 \ldots s_i$ in time $\mathcal{O}(i)$.
\end{lemma}

After such modification the algorithm works in linear time, which can be shown by defining a potential function depending just on the lengths of the snippets, see the original paper.

\section{Overview of the algorithm}

Our goal is to detect an occurrence of a pattern $p[1\twodots M]$ in a given text $t[1\twodots N]$, where $p$ and $t$ are described by a Lempel-Ziv-Welch parse of size $m$ and $n$, respectively. The difficulty here is rather clear: $M$ might be of order $m^{2}$, and hence looking at each possible prefix or suffix of the pattern would imply a quadratic (or higher) total complexity. As most efficient uncompressed pattern algorithms are
based on a more or less involved preprocessing concerning all prefixes or suffixes, such quadratic behavior seems difficult to avoid. Nevertheless,
we can try to use the following reasoning here: either the pattern is really complicated, and then $m$ is very similar to $M$, hence we can use the linear compressed pattern matching algorithm sketched in the previous section, or it is in some sense repetitive, and we can hope to speedup the preprocessing by building on this repetitiveness. In this section we give a high level formalization of this intuition. 

We will try to process whole codewords at once. To this aim we need the following technical lemma which allows us to compare large chunks of the text (or the pattern) in a single step. It follows from the linear time construction of the so-called {\it suffix tree of a tree}~\cite{ShibuyaTree} and constant time LCA queries~\cite{BenderLCA}.

\begin{lemma}
\label{lemma:parse preprocessing}
It is possible to preprocess in linear time a LZW parse of a text over an alphabet consisting of integers which can be sorted in linear time so that given any two codewords we can compute their longest common suffix in constant time.
\end{lemma}

We defer its proof to Section~\ref{section:lzw parse} as it is not really necessary to understand the whole idea. As an obvious corollary, given two codewords we can check if the shorter is a suffix of the longer in constant time.

In the very beginning we reverse both the pattern and the text. This is necessary because the above lemma tells how to compute the longest common suffix, and we would actually like to compute the longest common prefix. The only way we will access the characters of both the pattern and the text
is either through computing the longest common prefix of two reversed codewords, or retrieving a specified character of a reversed codeword (which can be performed in constant time using level ancestor queries), hence the input can be safely reversed without worrying that it will make working with it more complicated. We call those reversed codewords \emph{blocks}. Note that all suffixes of a block are valid blocks as well.

We start with classifying all possible patterns into two types. Note that this classification depends on both $m$ (size of the compressed pattern) and $n$ (size of the compressed texts) which might seem a little unintuitive.

\begin{definition}
A \emph{kernel} of the pattern is any (uncompressed) substring of length $n+m$ such that its border is at most $\frac{n+m}{2}$. A kernel decomposition of the pattern is its periodic prefix with period at most $\frac{n+m}{2}$ followed by a kernel.
\end{definition}

Note that the distance between two occurrences of such substring must be at least $\frac{n+m}{2}$, and hence a kernel occurs at most $2m$ times in the pattern and $2n$ times in the text. It might happen that there is no kernel, or in other words all relatively short fragments are highly repetitive. In such case the whole pattern turns out to be highly repetitive.

\begin{lemma}
\label{lemma:kernel}
The pattern either has a kernel decomposition or its period is at most $n+m$. Moreover, those two situations can be distinguished in linear time, and if a decomposition exists it can be found in the same complexity.
\end{lemma}

\begin{proof}
We start with decompressing the prefix of length $n+m$. If its period $d$ is at least $\frac{n+m}{2}$, we can return it as a kernel. Otherwise we
compute the longest prefix of the pattern and the pattern shifted by $d$ characters (or, in other words, we compute how far the period extends
in the whole pattern). This can be performed using at most $2(n+m)$ queries described in Lemma~\ref{lemma:parse preprocessing} (note that being able to check if one block is a prefix of another would be enough to get a linear total complexity here, as we can first identify the longest prefix consisting of whole blocks in both words and then inspect the remaining at most $\min(n,m)$ characters naively). If $d$ is the period
of the whole pattern, we are done. Otherwise we identified a substring $s$ of length $n+m$ followed by a character $a$ such that the period of
$s$ is at most $d\leq\frac{n+m-1}{2}$ but the period of $sa$ is different (larger). We remove the first character of $s$ and get $s'$. Let $d'$ be
the period of $s'a$. If $d'\geq\frac{n+m}{2}$, $s'a$ is a kernel. Otherwise $d,d'\leq\frac{n+m-1}{2}$ are both periods of $s'$, and hence by 
Lemma~\ref{lemma:periodicity} they are both multiplies of the period of $s'$. Let $b$ be the character such that $d$ is a period of $s'b$ (note that $a\neq b$). Because $d$ is a period of $s'b$, $s'[|s'|+1-d]=b$. Similarly, because $d'$ is a period of $s'a$, $s'[|s'|+1-d']=a$. Hence $s'[|s'|+1-d]\neq s'[|s'|+1-d']$, and because $(|s'|+1-d)-(|s'|+1-d')$ is a multiple of the period of $s'$ we get a contradiction. Note that the prefix before $s'a$ is periodic with period $d\leq\frac{n+m}{2}$ by the construction.
\end{proof}

If the pattern turns out to be repetitive, we try to apply the algorithm described in the preliminaries. The intuition is that while we required a certain preprocessing of the whole pattern, when its period is $d$ it is enough preprocess just its prefix of length $\mathcal{O}(d)$. This intuition is formalized in Section~\ref{section:periodic}. If the pattern has a kernel, we use it to identify $\mathcal{O}(n)$ potential occurrences, which we then manage to verify efficiently. The verification uses a similar idea to the one from Lemma~\ref{lemma:parse preprocessing} but unfortunately it turns out
that we need to somehow compress the pattern during the verification as to keep the running time linear. The details of this part are given in Section~\ref{section:kernel}.

\section{Detecting occurrence of a periodic pattern}
\label{section:periodic}

If the pattern is periodic, we would like to somehow use this periodicity so that we do not have to preprocess the whole pattern (i.e., build the suffix tree, LCA structure, compute the borders of all prefixes and suffixes, and so on). It seems reasonable that preprocessing just the first few repetitions of the period should be enough. More precisely, we will decompress a sufficiently long prefix of $p$ and compute some of its occurrences
inside the text. To compute those occurrences we apply a fairly simple modification of \proc{Levered-pattern-matching} called
\proc{Lazy-levered-pattern-matching}. 

First observe that both Lemma~\ref{lemma:concatenation occurrence} and Lemma~\ref{lemma:lever occurrence} can be modified in a straightforward way so that we get the leftmost occurrence, if any. The original procedure quits as soon it detects that the pattern occurs. We would like it to proceed
so that we get more than one occurrence, though. A naive solution would be to simply continue, but then the following situation could happen:
both $\ell$ and $|s_{k}|$ are very close to $m$, the pattern occurs both in the very beginning of $p[1\twodots\ell]s_{k}$ and somewhere close to the boundary between the two parts, and the longest suffix of the concatenation which is a prefix of the pattern is very short. Then we would detect just the first occurrence, and for some reasons that will be clear in the proof of Lemma~\ref{lemma:periodic} this is not enough. Hence whenever there is an occurrence in the concatenation, we skip just the first half of $p[1\twodots\ell]$ and continue, see lines~\ref{line:modify start}-\ref{line:modify end}. This is the only change in the algorithm.

\begin{figure*}
\begin{algorithm}[H]
\caption{\proc{Lazy-levered-pattern-matching}$(s_1, s_2, \ldots, s_n)$}
\begin{algorithmic}[1]
\State $\ell \gets $ longest prefix of $p$ ending $s_1$ \Comment{{\bf Lemma~\ref{lemma:longest suffix}}}
\State $k \gets 2$
\While{$k \leq n$ and $\ell+\sum_{i=k}^{n}|s_i|\geq m$} 
  \State choose $t\geq k$ minimizing $|s_k|+|s_{k+1}|+\ldots+|s_{t-1}|-\frac{|s_t|}{2}$
  \If{$\ell+|s_k|+|s_{k+1}|+\ldots+|s_{t-1}| \leq \frac{|s_t|}{2}$}
     \State output the first occurrence of $p$ in $p[1\twodots \ell] s_k s_{k+1} \ldots s_t$, if any \Comment{{\bf Lemma~\ref{lemma:lever occurrence}}}
     \State $\ell \gets$ longest prefix of $p$ ending $p[1\twodots \ell] s_k s_{k+1} \ldots s_t$ \label{line:aperiodic 1}\Comment{{\bf Lemma~\ref{lemma:lever prefix}}}
     \State $k \gets t+1$
  \Else
     \State output the first occurrence of $p$ in $p[1\twodots \ell] s_k$, if any \Comment{{\bf Lemma~\ref{lemma:concatenation occurrence}}} 
     \If{$p$ occurs in $p[1\twodots \ell] s_k$} \label{line:modify start}
       \State $\ell \gets$ longest prefix of $p$ ending $p[\left\lceil\frac{\ell}{2}\right\rceil\twodots \ell]$ \label{line:aperiodic 2}
       \State {\bf continue}
     \EndIf \label{line:modify end}

    \If{$p[1\twodots \ell]s_k$ is a prefix of $p$} \label{line:while after occurrence}
      \State $\ell \gets \ell + |s_k|$
      \State $k \gets k + 1$
      \State {\bf continue}
    \EndIf
    \State $b \gets$ longest long border of $p[1\twodots \ell]$ s.t. $p[1\twodots b] s_k$ is a prefix of $p$ \Comment{{\bf Lemma~\ref{lemma:extending long border}}}
    \If{$b$ is undefined}
      \State $\ell \gets$ longest prefix of $p$ ending $p[\left\lceil\frac{\ell}{2}\right\rceil\twodots \ell]$
      \State {\bf continue}
    \EndIf
    \State $\ell \gets b + |s_k|$
    \State $k \gets k + 1$ \label{line:while ends}
  \EndIf
\EndWhile
\end{algorithmic}
\end{algorithm}
\vspace{-0.5cm}
\end{figure*}

While \proc{Lazy-levered-pattern-matching} it is not capable of generating all occurrences in some cases, it will always detect a lot of them, in a certain sense. This is formalized in the following lemma.

\begin{lemma}
\label{lemma:aperiodic lots of occurrences}
If the pattern of length $m$ occurs starting at the $i$-th character, \proc{Lazy-levered-pattern-matching} detects at least one occurrence starting at the $j$-th character for some $j\in\{i-\frac{m}{2},i-\frac{m}{2}+1,\ldots,i\}$.
\end{lemma}

\begin{proof}
There are just two places where we can lose a potential occurrence: line~\ref{line:aperiodic 1} and~\ref{line:aperiodic 2}. More precisely,
it is possible that we output an occurrence and then skip a few others. We would like to prove that the occurrences we skip are quite close
to the occurrences we output. We consider the two problematic lines separately.
\begin{description}
\item[line~\ref{line:aperiodic 1}] $s_{t}$ is a lever, so $\ell+|s_{1}|+\ldots+|s_{t}|\leq\frac{3}{2}m$. Hence the distance between any two occurrences
of the pattern inside $p[1\twodots \ell] s_k s_{k+1} \ldots s_t$ is at most $\frac{m}{2}$. We output the first of them, and so can safely ignore the
remaining ones.
\item[line~\ref{line:aperiodic 2}] If there is an occurrence, we remove the first half of $p[1\twodots\ell]$ and might skip some other occurrences starting
there. If the first occurrence starts later, we will not skip anything. Otherwise we output the first occurrence starting in $p[1\twodots\frac{\ell}{2}]$,
and if there is any other occurrence starting there, their distance is at most $\frac{\ell}{2}\leq\frac{m}{2}$, hence we can safely ignore the latter.
\end{description}
\end{proof}


\begin{lemma}
\label{lemma:aperiodic time}
\proc{Lazy-levered-pattern-matching} can be implemented to work in time $\mathcal{O}(n)$ and use $\mathcal{O}(m)$ additional memory.
\end{lemma}

\begin{proof}
The proof is almost the same as in~\cite{GawrychowskiLZW}. The only difference as far as the running time is concerned is
line~\ref{line:aperiodic 2}. By removing the first half of $p[1\twodots\ell]$ we either decrease the current potential by $1$ or create a lever and thus can amortize the constant time used to locate the first occurrence of the pattern inside $p[1\twodots\ell] s_{k}$.
\end{proof}

Note that \proc{Lazy-levered-pattern-matching} works with a sequence of snippets. By first applying the preprocessing mentioned in
the preliminaries we can use it to compute a small set which approximates all occurrences in a compressed text.

\begin{lemma}\label{lemma:set of occurrences}
\proc{Lazy-levered-pattern-matching} can be used to compute a set $S$ of $\mathcal{O}(n)$ occurrences of an uncompressed pattern of length $m\geq n$ in a compressed text such that whenever there is an occurrence starting at the $i$-th character, $S$ contains $j$ from $\{i-\frac{m}{2},i-\frac{m}{2}+1,\ldots,i\}$.
\end{lemma}

\begin{proof}
As mentioned in the preliminaries, we can reduce compressed pattern matching to pattern matching in a sequence of snippets in linear time.
Because $m\geq n$, the preprocessing does not produce any occurrences yet. Then we apply \proc{Lazy-levered-pattern-matching}. Because its
running time is linear by Lemma~\ref{lemma:aperiodic time}, it cannot find more than $\mathcal{O}(n+m)$ occurrences. A closer look
at the analysis shows that the number of occurrences produced can be bounded by the potentials of all sequences created during
the initial preprocessing phase, which as shown in~\cite{GawrychowskiLZW} is at most $\mathcal{O}(n)$.
\end{proof}

Now it turns out that if the pattern is compressed but highly periodic, the occurrences found in linear time by the above lemma applied to
a sufficiently long prefix of $p$ are enough to detect an occurrence of the whole pattern.

\begin{lemma}
\label{lemma:periodic}
Fully compressed pattern matching can be solved in linear time if the pattern is of compressed size $m\geq n$ and its period is at most $\frac{n+m}{2}$. Furthermore, given a set of $r$ potential occurrences we can verify all of them in $\mathcal{O}(n+m+r)$ time.
\end{lemma}

\begin{proof}
We build the shortest prefix $p[1\twodots\alpha d]$ such that $\alpha d\geq n+m$, where $d\leq\frac{n+m}{2}$ is the period of the whole pattern. Observe that $\alpha d\leq \frac{3}{2}(n+m)$ and hence we can afford to store this prefix in an uncompressed form. By Lemma~\ref{lemma:set of occurrences} we construct a set $S$ of $\mathcal{O}(n)$ occurrences of $p[1\twodots\alpha d]$ such that for
any other occurrence starting at the $i$-th character there exists $j\in S$ such that $0\leq i-j\leq\frac{\alpha d}{2}\leq \frac{3}{2}(n+m)$. We partition the elements in $S$ according to their remainders modulo $d$ so that $S_{r}=\{j\in S: j\equiv r \pmod d\}$ and consider each $S_{r}$ separately. Note that we can easily ensure that its elements are sorted by either applying radix sort to the whole $S$ or observing that \proc{Lazy-levered-pattern-matching} generate the occurrences from left to right.

We split $S_{r}$ into maximal groups of consecutive elements $x_{1}<x_{2}<\ldots <x_{k}$ such that $x_{i+1}\leq x_{i}+\frac{\alpha d}{2}$, which clearly can be performed in linear time with a single left-to-right sweep. Each such group actually corresponds to a fragment starting at the $x_{1}$-th character and ending at the $x_{k}+\alpha d-1$ character which is a power of $p[1\twodots d]$. This is almost enough to detect an occurrence of the whole pattern. If the fragment is sufficiently long, we get an occurrence. In some cases this is not enough to detect the occurrence because we might be required to extend the period to the right as to make sufficient space for the whole pattern. Fortunately, it is impossible to repeat $p[1\twodots d]$ more than $\frac{3}{2}\alpha$ times starting at the $x_{k}$ character, as otherwise we would have another $x_{k+1}\in S_{r}$ which we might have used to extend the group. Hence to compute how far the period extends it would be enough to align $p[1\twodots\alpha d]p[1\twodots\frac{\alpha d}{2}]$ starting at the $x_{k}$ character and compute the first mismatch with the text. We can assume that all suffixes of $p[1\twodots\alpha d]$ are blocks with just a linear increase in the problem size, and hence we can apply Lemma~\ref{lemma:parse preprocessing} to preprocess the input so that each such alignment can be processed in time proportional to the number of block in the corresponding fragment of the text. To finish the proof, note that any single block in the text will be processed at most twice. Otherwise we would have two groups ending at the $x_{k}$-th and $x'_{k'}$-th characters such that $\left|x_{k}+\alpha d - (x'_{k'}+\alpha d)\right|\leq\frac{\alpha d}{2}$ and that would mean that one of those groups is not maximal. After computing how far the period extends after each group, we only have to check a simple arithmetic conditions to find out if the pattern occurs starting at the corresponding $x_{1}$.

To verify a set of $r$ potential occurrences, we construct the groups and compute how far the period extends after each of them as above. Then for each potential occurrence starting at the $b_{i}$-th character 
we lookup the corresponding $S_{b_{i}\bmod d}$ and find the rightmost group such that $x_{1}\leq b_{i}$.  We can verify an occurrence by
looking up how far the period extends after the $x_{k}$-th character and checking a simple arithmetic condition. To get the claimed time bound, observe that we do not have to perform the lookup separately for each possible occurrence. By first splitting them according to their remainders modulo $d$ and sorting all $x_{1}$ and $b_{i}$ in linear time using radix sort consisting of two passes we get a linear total complexity.
\end{proof}

\section{Using kernel to accelerate pattern matching}
\label{section:kernel}

We start with computing all occurrences of the kernel in both the pattern and the text. Because the kernel is long and aperiodic, there are no more than $2m$ of the former and $2n$ of the latter. The question is if we are able to detect all those occurrences efficiently. It turns out
that because the kernel is aperiodic, \proc{Lazy-levered-pattern-matching} can be (again) used for the task. More formally, we have the following
lemma.

\begin{lemma}\label{lemma:aperiodic}
\proc{Lazy-levered-pattern-matching} can be used to compute in $\mathcal{O}(n+m)$ time all occurrences of an aperiodic pattern of length $m\geq n$ in a compressed text. 
\end{lemma}

\begin{proof}
By Lemma~\ref{lemma:set of occurrences} we can construct in linear time a set of occurrences such that any other occurrence is quite close to one of them. But the pattern is aperiodic, so if it occurs at positions $i$ and $j$ with $|i-j|\leq\frac{m}{2}$, then in fact $i=j$. Hence the set contains all occurrences.
\end{proof}

We apply the above lemma to find the occurrences of the kernel in both the pattern and the text. Each occurrence of the kernel in the text gives us a 
possible candidate for an occurrence of the whole pattern (for example by aligning it with the first occurrence of the kernel in the pattern). Hence
we have just a linear number of candidates to verify. Still, the verification is not trivial. An obvious approach would be to repeat a computation
similar to the one from Lemma~\ref{lemma:kernel} for each candidate. This would be too slow, though, as it might turn out that some blocks from the pattern are inspected multiple times. We require a slightly more sophisticated approach.


Using (any) kernel decomposition of the pattern we represent it as $p=p_1 p_2 p_3$, where the period of $p_1$ is at most $\frac{n+m}{2}$, and
$p_2$ is a kernel. We start with locating all occurrences of $p_2 p_3$ in the text. It turns out that because $p_2$ is aperiodic, there cannot
be too many of them. Hence we can afford to generate all such occurrences and then verify if any of them is preceded by $p_1$ as follows:

\begin{enumerate}
\item if $|p_1|\geq n$ then we can directly apply Lemma~\ref{lemma:periodic},
\item if $|p_1|<n$ then take the prefix of $p_1 p_2$ consisting of the first $n+m$ letters. Depending on whether this prefix is periodic with the period at most $\frac{n+m}{2}$ or aperiodic, we can apply Lemma~\ref{lemma:periodic} or Lemma~\ref{lemma:aperiodic}.
\end{enumerate}

The most involved part is computing all occurrences of $p_2 p_3$. To find them we construct a new string $T$ by concatenating the suffix of the pattern and the text:
$$T=p_2 p_3 \$t[1\twodots N]$$
For this new string we compute the values of the prefix function defined in the following way:
$$
\PREF[i] = \max\{j : T[k]=T[i+k-1] \text{ for all } k=1,2,\ldots,j \}
$$
Of course we cannot afford to compute $\PREF[i]$ for all possible $N+M$ values of $i$. Fortunately, $\PREF[i]\geq |p_2|$ iff $p_2$ occurs in $T$ 
starting at the $i$-th character. Because $|p_2|=n+m$ and $p_2$ is aperiodic, there are no more than $2\frac{N+M}{n+m}\leq n+m$ such values of $i$. We aim
to compute $\PREF[i]$ just for those $i$. First lets take a look at the relatively well-known algorithm which computes all $\PREF[i]$ for
all $i$, which can be found in the classic stringology book by Crochemore and Rytter~\cite{Jewels}. We state its code for the sake of completeness. 
$\proc{Naive-scan}(x,y)$ performs a naive scanning
of the input starting at the $x$-th and $y$-th characters. $\proc{PREF}$ uses this procedure in a clever way as to reuse already processed
parts of the input and keep the total running time linear. The complexity is linear because the value of $s+\PREF[s]$ cannot decrease nor
exceed $|T|$, and whenever it increases we are able to pay for the time spent in \proc{Naive-scan} using the difference between the new and the old
value.

\begin{figure*}
\begin{minipage}{\textwidth}
\begin{algorithm}[H]
\caption{\proc{PREF}$(T[1\twodots |T|])$}
\begin{algorithmic}[1]
\State $\PREF[1]=0$, $s \gets 1$
\For{$i=2,3,\ldots,|T|$}\label{line:for starts}
  \State $k \gets i-s+1$
  \State $r \gets s+PREF[s]-1$
  \If{$r<i$}
    \State $\PREF[i] = \proc{Naive-Scan}(i,1)$
    \If{$\PREF[i] > 0$}
      \State $s \gets i$
    \EndIf
  \ElsIf{$\PREF[k]+k < \PREF[s]$}
    \State $\PREF[i] \gets \PREF[k]$
  \Else
    \State $x \gets \proc{Naive-scan}(r+1,r-i+2)$ \label{line:access}
    \State $\PREF[i] \gets r-i+1+x$
    \State $s \gets i$
  \EndIf
\EndFor
\State $\PREF[1] = |t|$
\end{algorithmic}
\end{algorithm}
\end{minipage}
\end{figure*}

We will transform this algorithm so that it computes only $\PREF[i]$ such that the kernel occurs starting at the $i$-th character. We call
such positions $i$ \emph{interesting}. The first problem we encounter is that we need a constant time access to any $\PREF[i]$ and cannot afford to allocate a table of length $|T|$. This can be easily overcome.

\begin{lemma}
\label{lemma:random access}
A random access table $\PREF$ such that $\PREF[i]>0$ iff the kernel occurs starting at the $i$-th character can be implemented in constant time
per operation requiring space and preprocessing time not exceeding the compressed size of $T$, which is $\mathcal{O}(n+m)$.
\end{lemma}

\begin{proof}
Observe that any two occurrences of the kernel cannot be too close. More precisely, their distance must be at least $\frac{n+m}{2}$. We split
the whole $T$ into disjoint fragments of size $\frac{n+m}{2}$. There are no more than $2(n+m)$ of them and there is at most one occurrence in
each of them. Hence we can implement the table by allocating an array of size $2(n+m)$ with each entry storing at most one element.
\end{proof}

We modify line~\ref{line:for starts} so that it iterates only through $i$ which are interesting. Note that whenever we access some $\PREF[j]$ inside,
$j$ is either $i, s$ or $k=i-s+1$. In the first two cases it is clear that the corresponding positions are interesting so we can access the corresponding
value using Lemma~\ref{lemma:random access}. The third case is not that obvious, though. It might happen that $k$ is not interesting and we will
get $\PREF[k]=0$ instead of the true value. If $r\geq i+|p_{2}|-1$ then because $p_{2}$ occurs at $i$, it occurs at $k$ as well, and so $k$ is interesting.
Otherwise we cannot access the true value of $\PREF[k]$, so we start a naive scan by calling $\proc{Naive-scan}(i+|p_{2}|,|p_{2}|+1)$ (we can start
at the $|p_{2}|+1$-th character because $p_{2}$ occurs at $i$). After the scanning we set $s\gets i$. Note that because $r<i+|p_{2}|-1$, this
increases the current value of $s+\PREF[s]$, and we can use the increase to amortize the scanning.

We still have to show how to modify \proc{Naive-scan}. Clearly we cannot afford to perform the comparisons character by character. By the increasing
$s+\PREF[s]$ argument, any single character from the text is inspected at most once by accessing $T[x]$ (we call it a left side access). It might be inspected multiple  times by accessing $T[y]$, though (which we call a right side access). We would like to perform the comparisons block by block using Lemma~\ref{lemma:parse preprocessing}. After a single query we skip at least one block. If we skip a block responsible 
for the left side access, we can clearly afford to pay for the comparison. We need to somehow amortize the situation when we skip a block
responsible for the right side access. For this we will iteratively compress the input (this is similar to the idea used in~\cite{GasieniecFully} with the exception that we work with $\PREF$ instead of the failure function). 
More formally, consider the sequence of blocks describing $p_{2} p_{3}$. First note that no further blocks from $T$ will be responsible for a right side access because of the unique $\$$ character. Whenever some two neighboring blocks $b_{1}, b_{2}$ from this prefix occur in the same block from the text $b'$, we would like to glue them, i.e., replace by a single block. We cannot be sure that there exists a block corresponding to their concatenation, but because we know where it occurs in $b'$ we can extract (in constant time, by using the level ancestor data structure~\cite{BenderAncestor} to preprocess the whole code trie) a block for which the concatenation is a prefix. We will perform such replacement whenever possible. Unfortunately, after such replacement we face a new problem: $p_2 p_3$ is represented as a concatenation of prefixes of blocks instead of whole blocks. Nevertheless, we can still apply Lemma~\ref{lemma:parse preprocessing} to compute the longest common prefix of two prefixes
of blocks $b[1..i]$ and $b'[1..i']$ by first computing the longest common prefix of $b$ and $b'$, and decreasing it if it exceeds $\min(i,i')$.
More formally, we store a \emph{block cover} of $p_2 p_3$.

\begin{definition}
A block cover of a word $w$ is a sequence $b_1[1\twodots i_1], \ldots, b_k[1\twodots i_k]$ of prefixes of blocks 
such that their concatenation is equal to $w$.
\end{definition}

\begin{figure*}
\begin{minipage}[b]{0.5\linewidth}
\centering
\includegraphics[scale=0.7]{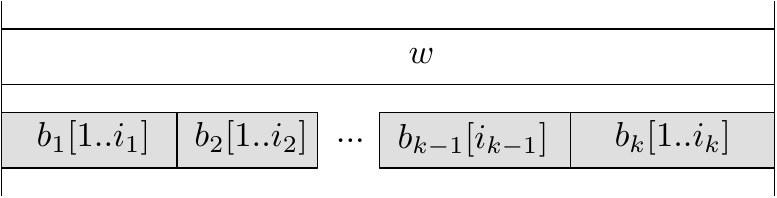}
\caption{A block cover of $w$.}
\label{figure:codeword cover}
\end{minipage}
\begin{minipage}[b]{0.5\linewidth}
\centering
\includegraphics[scale=0.7]{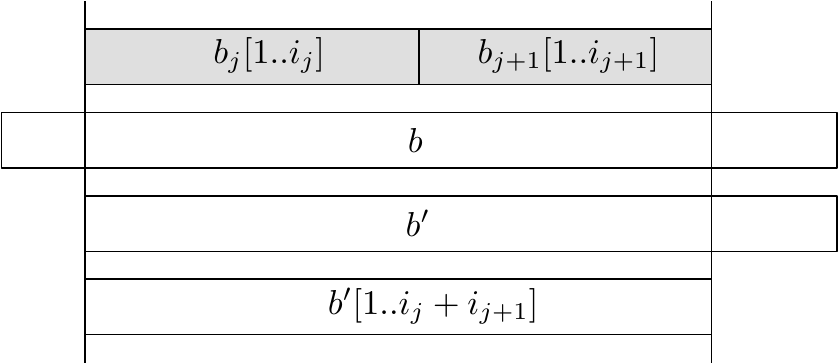}
\caption{Compressing the current block cover.}
\label{figure:compressing cover}
\end{minipage}
\end{figure*}

This definition is illustrated on Figure~\ref{figure:codeword cover}. Obviously, the initial partition of $p_2 p_3$ into blocks is a valid block cover.
If during the execution of \proc{Naive-scan} we find out that two neighboring elements $b_j[1\twodots i_1], b_{j+1}[1\twodots i_{j+1}]$ of the current cover occur in some other longer block $b$, we replace them with the corresponding prefix of $b'$, see Figure~\ref{figure:compressing cover}. 

We store all $b_j[1\twodots i_j]$ on a doubly linked list and update its element accordingly after each replacement. The final required step
is to show how we can quickly access in line~\ref{line:access} the block corresponding to $r-i+2$. We clearly are allowed to spend just constant
time there. We keep a pointer to the block covering the $r$-th character of the pattern. Whenever we need to access the block covering
the $(r-i+2)$-th character, we simply move the pointer to the left, and whenever the current longest match extends, we move the pointer to the right.
We cannot move to the left more time than we move to the right, and the latter can be bounded by the number of blocks in the whole $p_1 p_2$ if we replace the neighboring blocks whenever it is possible.


\begin{lemma}
\label{lemma:kernel acceleration}
Fully compressed pattern matching can be solved in linear time if we are given the kernel of the pattern.
\end{lemma}


\begin{theorem}
Fully LZW-compressed pattern matching for strings over a polynomial size integer alphabet can be solved in optimal linear time assuming the word RAM model.
\end{theorem}

\section{LZW parse preprocessing}
\label{section:lzw parse}

The goal of this section is to prove Lemma~\ref{lemma:parse preprocessing}. We aim to preprocess the codewords trie so that
given any two codewords, we can compute their longest common suffix in constant time. It turns out that we can use some
existing (while maybe not very known) tools to achieve that. The \emph{suffix tree of a tree} $A$, where $A$ is a tree with edges
labeled with single characters, is defined as the compressed trie containing $s_A(v)\$$ for all $v\in A$, where $s_A(v)$ is the string
constructed by concatenating the labels of all edges on the $v$-to-root path in $A$, see Figure~\ref{figure:suffix tree of tree}. This has been
first used by Kosaraju~\cite{KosarajuTree}, who developed a $\mathcal{O}(|A|\log |A|)$ time construction algorithm, where $|A|$ is the number
of nodes of $A$. The complexity has been then improved by Breslauer~\cite{BreslauerTree} to just $\mathcal{O}(|A|\log|\Sigma|)$ (which for constant 
alphabets is linear), and by Shibuya~\cite{ShibuyaTree} to linear for integer alphabets.

\begin{figure*}
\centering
\includegraphics{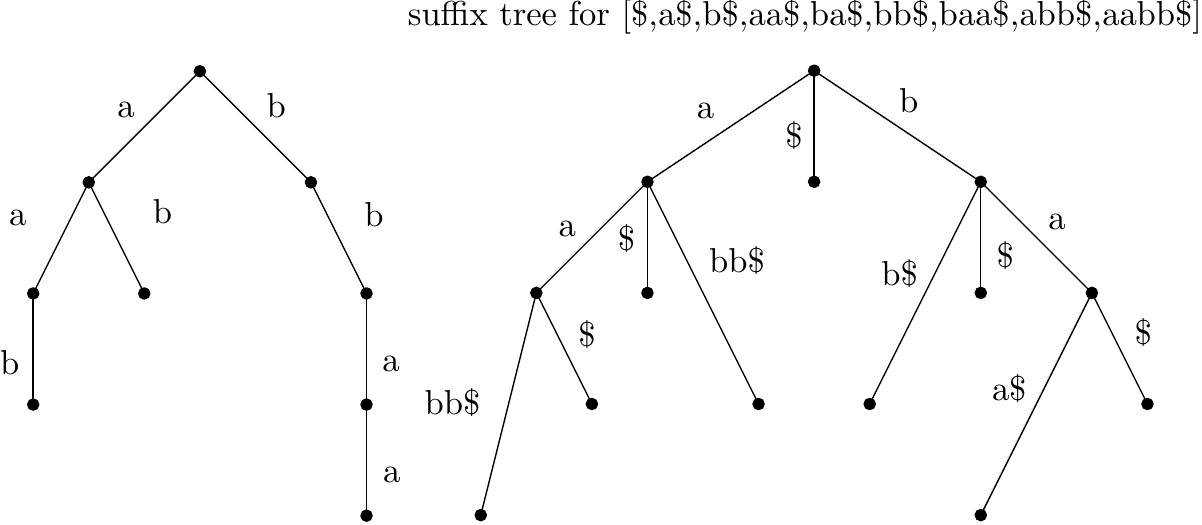}
\caption{A trie (on the left) and its suffix tree (on the right).}
\label{figure:suffix tree of tree}
\end{figure*}

We build the suffix tree of the codeword trie $T$ in linear time~\cite{ShibuyaTree}. As a result we also get for any node $v$ of the input trie the node of the suffix tree corresponding to $s_T(v)\$$. Now assume that we would like to compute the longest common suffix of two codewords corresponding to
nodes $u$ and $v$ in the input trie. In other words, we would like to compute the longest common prefix of $s_T(u)$ and $s_T(v)$. This can be found
in constant time after a linear time preprocessing by retrieving the lowest common ancestor of their corresponding nodes in the suffix tree~\cite{BenderLCA}. 


\bibliographystyle{abbrv}
\bibliography{biblio}

\end{document}